\documentclass[11pt]{article}
\usepackage{amsmath,amsthm,bm,amssymb,fullpage,times,url}
\usepackage{booktabs}
\usepackage{graphicx}
\usepackage[boxed,noend]{algorithm2e}
\usepackage{tikz}

\makeatletter
\def\mop#1{\mathop{\operator@font {#1\null}}}
\makeatother
\def\GL{\mop{GL}}
\def\HW{\mop{weight}}
\def\bigO{\mop{O}}

\newcommand{\ket}[1]{\mid\hspace*{-4pt}{#1}\rangle}

\newtheorem{example}{{\bf Example}}[section]
\newtheorem{remark}{{\bf Remark}}[section]
\newtheorem{proposition}{{\bf Proposition}}[section]
\newtheorem{theorem}{{\bf Theorem}}[section]

\newtheorem{corollary}{{\bf Corollary}}[section]

\begin{document}
\title{Automatic synthesis of quantum circuits for point addition on ordinary binary elliptic curves}
\author{Parshuram Budhathoki\rule{10ex}{0ex}Rainer Steinwandt\\\\
Florida Atlantic University\\
Department of Mathematical Sciences\\
Boca Raton, FL 33431\\
\texttt{$\{$pbudhath,rsteinwa$\}$@fau.edu}}
\date{}
\maketitle
\begin{abstract}
  Implementing the group arithmetic is a cost-critical task when designing quantum circuits for Shor's algorithm to solve the discrete logarithm problem. We introduce a tool for the automatic generation of addition circuits for ordinary binary elliptic curves, a prominent platform group for digital signatures. Our \texttt{Python} software generates circuit descriptions that, without increasing the number of qubits or $T$-depth, involve less than 39\% of the number of $T$-gates in the best previous construction. The software also optimizes the (CNOT) depth for ${\mathbb F}_2$-linear operations by means of suitable graph colorings.
\end{abstract}

\section{Introduction}
Ordinary binary elliptic curves are an algebraic structure of great cryptographic significance. All binary curves suggested in the Digital Signature Standard \cite{FIPS1864} fall in this class, and the cost of implementing Shor's quantum algorithm \cite{Sho97} in such groups has been explored by various authors. While optimizing the implementation of the Quantum Fourier Transform is a quite well understood task, minimizing the implementation cost of the scalar multiplication in Shor's algorithm remains a design challenge.

Approaches by Kaye and Zalka \cite{KaZa04} and by Maslov et al. \cite{MMCP09b} rely on the use of projective coordinates and efficient circuits for adding fixed (classically precomputed) points in a right-to-left variant of the double-and-add algorithm. In fact, only a ``generic'' addition of a fixed point is implemented, avoiding a handling of special cases of the addition law (doubling a point, adding a point with its inverse, or with the identity element). As observed in \cite{MMCP09b}, it is sufficient to represent the input and output points of such a point addition circuit with projective coordinates.
Amento et al. \cite{ARS12} suggest to replace ordinary projective coordinates with a representation used by Higuchi and Takagi \cite{HiTa00}, therewith reducing the number of $T$-gates\footnote{As is common, we do not distinguish between $T$- and $T^\dagger$-gates in statements on the number of $T$-gates or the $T$-depth.} needed.
Taking the number of $T$-gates as cost measure, this is so far the most efficient implementation proposed, but as noted in \cite{RoSt13}, alternative constructions can reduce the design complexity: If dedicated circuitry for doubling a point is available, scalar multiplication can be realized by invoking only two types of addition circuits---rather than several hundred different ones when dealing with cryptographically significant parameters. These doubling circuits do impact the gate count, however. Happily, with the software tool presented below, designing addition circuits can be automated, making the derivation of a few hundred addition circuits for different points a realistic option.

To optimize circuit depth, \cite{RoSt13} suggest a tree-style organization of the scalar multiplication in Shor's algorithm. However, this method builds on general addition circuits for the elliptic curve, i.\,e., addition circuits which have two variable input points and handle all cases of the addition law. In \cite{RoSt13} complete binary Edwards curves \cite{BLF08} are used for this purpose. While the resulting circuit depth is compelling, the number of $T$-gates and number of qubits is much worse than with a right-to-left double and add procedure. When aiming at a small $T$-gate count, optimizing quantum circuits for the ``generic'' addition of a fixed point appears to be the more preferrable research direction. In this paper, the central optimization criteria is the number of $T$-gates, and as secondary criteria we take the $T$-depth and the number of qubits into account.

\paragraph{Contribution.} Building on an addition formula by Al-Daoud et al. \cite{AMRK02} we show that the number of $T$-gates in the best available circuit to add a fixed point \cite{ARS12b} can be reduced by more than 60\% without affecting the $T$-depth negatively. At the same time the number of qubits can be reduced at the cost of a depth increase of about $4n$ when working with curves over $\mathbb F_{2^n}$. The circuit descriptions are derived automatically, and by means of edge colorings of certain bipartite graphs it is ensured that the involved subcircuits for ${\mathbb F}_2$-linear operations---such as multiplication by a constant, squaring and computing a square root---are optimized. For parameters of interest the latter allow substantial savings in the number of CNOT gates compared to the bounds used in \cite{ARS12b}. Building on an available polynomial-basis arithmetic for the underlying binary field, the {\tt Python}~\cite{Pythonreference} software we introduce synthesizes for a given curve and curve point an optimized addition circuit and outputs this circuit as a {\tt .qc} file. This file can then be processed with {\tt QCViewer} \cite{qcreference}, for instance, or more generally serve as input for automated or manual post-processing.

\paragraph{Structure of this paper.} In the next section we look at the choice of a suitable (polynomial-basis) representation of the underlying finite field and show how edge colorings can be used to find efficient circuits for squaring, constant multiplication, and square root computation. In Section~\ref{sec:ellipticcurve} we combine such circuits with Al-Daoud et al.'s addition formula for ordinary binary elliptic curves to derive a new quantum circuit for point addition with improved $T$-gate complexity. Complementing the theoretical discussion, we discuss concrete examples of circuits that have been synthesized with our software.

\section{Quantum circuits for ${\mathbb F}_{2^n}$-arithmetic}\label{sec:Finitefield}
A binary field ${\mathbb F} _{2^n}$ can be represented in various different ways, resulting in potentially very different quantum circuits to realize the arithmetic. The use of a normal basis has been considered \cite{ARS12}, but for elliptic curve addition with a small $T$-gate complexity, a polynomial basis representation seems the preferrable choice (see the discussion in \cite[Section~2]{ARS12b}). In a polynomial basis representation, $\mathbb{F}_{2^n}$ is expressed as a quotient
$$\mathbb{F}_{2^n} = {\mathbb{F}_{2}[x]}/(p)$$
of the univariate polynomial ring ${\mathbb F}_2[x]$ with binary coefficients, where $p \in \mathbb{F}_{2^n}[x]$ is an irreducible polynomial of degree $n$. Having fixed $p$, each element $a\in \mathbb{F}_{2^n}$ is uniquely represented by a bit vector $(a_0, a_1,\dots, a_{n-1})\in{\mathbb F}_2^n$ such that $a = a_0+a_1x+\dots+a_{n-1}x^{n-1}\pmod p$. This bit vector is naturally represented with $n$ qubits $\ket{a_0}\dots\ket{a_{n-1}}$. To implement point addition with a projective representation on a binary elliptic curve, we rely on addition, multiplication, multiplication with a non-zero constant and squaring in the underlying finite field. Quantum circuits for these tasks are available (cf. \cite{BBF03,KaZa04,MMCP09b,ARS12b}):

\begin{description}
\item[Addition:] To add two field elements $a, b\in{\mathbb F}_{2^n}$, we can simply use $n$ CNOT gates that operate in parallel: $$\ket{a, b} \longmapsto\ket{a,a \oplus b}.$$ Alternatively, if the operands are to remain unchanged, we can implement $$\ket{a,b}\ket{0^n} \longmapsto\ket{a,b}\ket{a \oplus b}$$ in the obvious way with $2n$ CNOT gates in depth $2$.
\item[Multiplication:] Optimizing the field multiplier is outside the scope of this paper, and subsequently we will use a linear-depth construction by Maslov et al. \cite{MMCP09b}. With this method one can multiply two elements $a, b\in{\mathbb F}_{2^n}$ with no more than $n^2$ Toffoli gates and $n^2-1$ CNOT gates. For certain choices of $p$, including trinomials, this bound can be improved further. We note that the point addition circuit developed in Section~\ref{sec:ellipticcurve} treats the underlying ${\mathbb F}_{2^n}$-multiplier as a black box. If more efficient field multipliers become available, integrating these into our synthesis tool should be straightforward.
\item[Multiplication with a non-zero constant and squaring:] Both of these operations are linear, and \cite{ARS12b} argue that an LUP decompositon yields a circuit of depth $\le 2n$ that can be realized with $n^2+n$ CNOT gates. Although no $T$-gates are needed for these operations, optimizing this step further is worthwhile: for the binary elliptic curves in the Digital Signature Standard \cite{FIPS1864}, we have $n\ge 163$ and accordingly the complete scalar multiplicaton in Shor's algorithm involves several hundred addition circuits.
\end{description}

\subsection{Optimizing ${\mathbb F}_2$-linear operations and minimal edge colorings}
Multplication by a constant and squaring are special cases of finding a quantum circuit implementing a map
$$\ket{a}\ket{0^n}\longmapsto\ket{a}\ket{b_0\dots b_{n-1}}$$
where $a=\sum_{i=0}^{n-1}a_ix^i+(p)$ is an arbitrary input from ${\mathbb F}_{2^n}$ and $(b_0,\dots,b_{n-1})=(a_0,\dots,a_{n-1})\cdot M$ for some non-singular matrix $M\in\GL_n({\mathbb F}_2)$. Obviously such a vector-by-matrix multiplication can be implemented with one CNOT gate for each non-zero entry of $M$. This can be done without ancillae qubits using a total of $\HW(M)$ CNOT gates. To minimize the circuit depth we interpret $M=(m_{i,j})_{0\le i<n}$ as biadjacency matrix of a bipartite graph. Namely, the graph associated with $M$ has $2n$ vertices with the vertex set splitting into the ``control part'' $\{a_0,\dots,a_{n-1}\}$ and the ``target part'' $\{b_0,\dots,b_{n-1}\}$. Each CNOT corresponds to exactly one edge: there is an edge between $a_i$ and $b_j$ if and only if $m_{i,j}=1$. An edge coloring of this graph with $d$ colors immediately yields a quantum circuit to multiply by $M$ in depth $d$---all CNOT gates corresponding to an edge of the same color operate on disjoint qubits and therewith can be executed in parallel.

The minimal possible value of $d$ is known as \emph{chromatic index} of the graph. For a bipartite graph the chromatic index is equal to the maximum degree of a vertex, i.\,e., equal to the maximal Hamming weight of the rows and columns of $M$. Efficient classical algorithms for finding such a minimal edge coloring are known (see, e.\,g., \cite{COS01}). For our software implementation we use a solution by Pointdexter \cite{Poi06} to find the required edge colorings.

\begin{proposition}\label{prop:lowdepth}
   Multiplication by a matrix $M\in\GL_n({\mathbb F}_2)$, i.\,e., the map $\ket{u}\ket{v}\longrightarrow \ket{u}\ket{v+M\cdot u}$ with arbitrary input vectors $u,v\in{\mathbb F}_2^n$, can be implemented with $\HW(M)$ CNOT gates. For this, an ancillae-free circuit of depth equal to the maximal Hamming weight of the rows and colums of $M$ is sufficient.
\end{proposition}
As worst-case bounds this implies the following.
\begin{corollary}
    Multiplication by an arbitrary matrix $M\in\GL_n(\mathbb F_2)$ can be implemented with at most $n^2-n+1$ CNOT gates, using an ancillae-free circuit of depth at most $n$.
\begin{proof}
   Because of Proposition~\ref{prop:lowdepth} it suffices to show that $\HW(M)\le n^2-n+1$. Suppose this is not true, i.\,e., $\HW(M)\ge n^2-n+2$. Then $M$ must contain at least two rows with all entries being equal to $1$, as having $n-1$ rows each of weight $\le n-1$ results in a matrix of weight $\le n+(n-1)\cdot(n-1)=n^2-n+1$. To bring the weight to $n^2-n+2$ at least one more row must be completed to an all-one row. Thus $M$ has two identical rows, which contradicts $M\in\GL_n(\mathbb F_2)$.
\end{proof}
\end{corollary}

\begin{example}[Constant multiplication in ${\mathbb F}_8$] Consider $n=3$ and $p=1+x+x^3$, i.\,e., ${\mathbb F}_{2^3}={\mathbb F}_2[x]/(1+x+x^3)$. Then multiplying an arbitrary polynomial $a=a_0+a_1x+a_2x^2+(p)$ with $1+x+x^2+(p)$ can be interpreted as multiplying the coefficient vector $(a_0, a_1, a_2)$ with the following matrix of weight $6$:
$$\left(\begin{array}{ccc}
  1&1&1\\
  1&0&1\\
  1&0&0\\
  \end{array}
  \right)$$
From this matrix we obtain the subsequent graph with six vertices and six edges.

\begin{center}
\begin{tikzpicture}[every node/.style={circle,draw}]
  \node at (0,0) (a2) {$a_2$};
  \node at (0,1) (a1) {$a_1$};
  \node at (0,2) (a0) {$a_0$};
  \node at (3,0) (b2) {$b_2$};
  \node at (3,1) (b1) {$b_1$};
  \node at (3,2) (b0) {$b_0$};

  \path[dotted] (a0) edge (b0);
  \path[dashed] (a1) edge (b0);
  \path (a2) edge (b0);

  \path (a0) edge (b1);
  \path[dashed] (a0) edge (b2);
  \path (a1) edge (b2);
\end{tikzpicture}
\end{center}

Consequently, we need a total of six CNOT gates, and in accordance with the matrix containing a row (and a column) of weight three, the graph has chromatic index $3$, yielding the quantum circuit shown in Figure~\ref{fig:Mulbycons}. The first three CNOT gates can be executed simultaneously (solid edges), similarly the next two CNOT gates can be applied at the same time (dashed edges), and finally the last CNOT gate can be applied (dotted edge).
\begin{figure}[h]
\centering
  \includegraphics[scale=0.3]{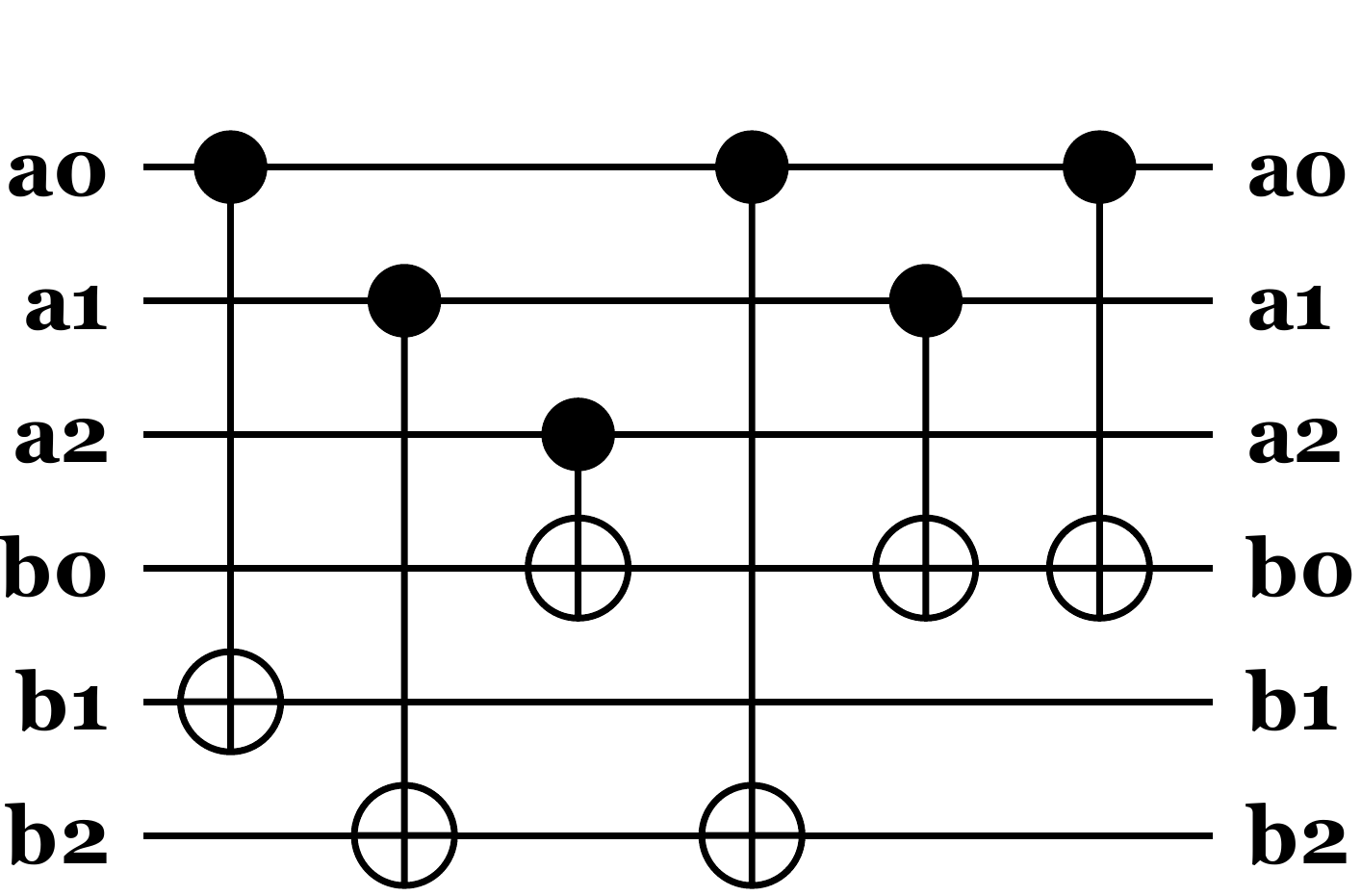}
  \caption{A circuit for ancillae-free multiplication of $a\in \mathbb{F}_{2}[x]/(1+x+x^3)$ with $1+x+x^2+(1+x+x^3)$.\label{fig:Mulbycons}}
\end{figure}
\end{example}

\begin{example}[Squaring in ${\mathbb F}_{128}$] Now let $n=7$ and choose $p=1+x+x^7$. Squaring $a_0+a_1x+\dots+a_6x^6+(p)\in{\mathbb F}_2[x]/(p)$ can be expressed as multiplying the coefficient vector $(a_0,\dots,a_6)\in{\mathbb F}_2^7$ by the matrix
$$\left(\begin{array}{ccccccc}
   1&0&0&0&0&0&0\\
   0&0&1&0&0&0&0\\
   0&0&0&0&1&0&0\\
   0&0&0&0&0&0&1\\
   0&1&1&0&0&0&0\\
   0&0&0&1&1&0&0\\
   0&0&0&0&0&1&1
  \end{array}\right).
$$

\begin{figure}[tbh]
\centering
\includegraphics[scale=.3]{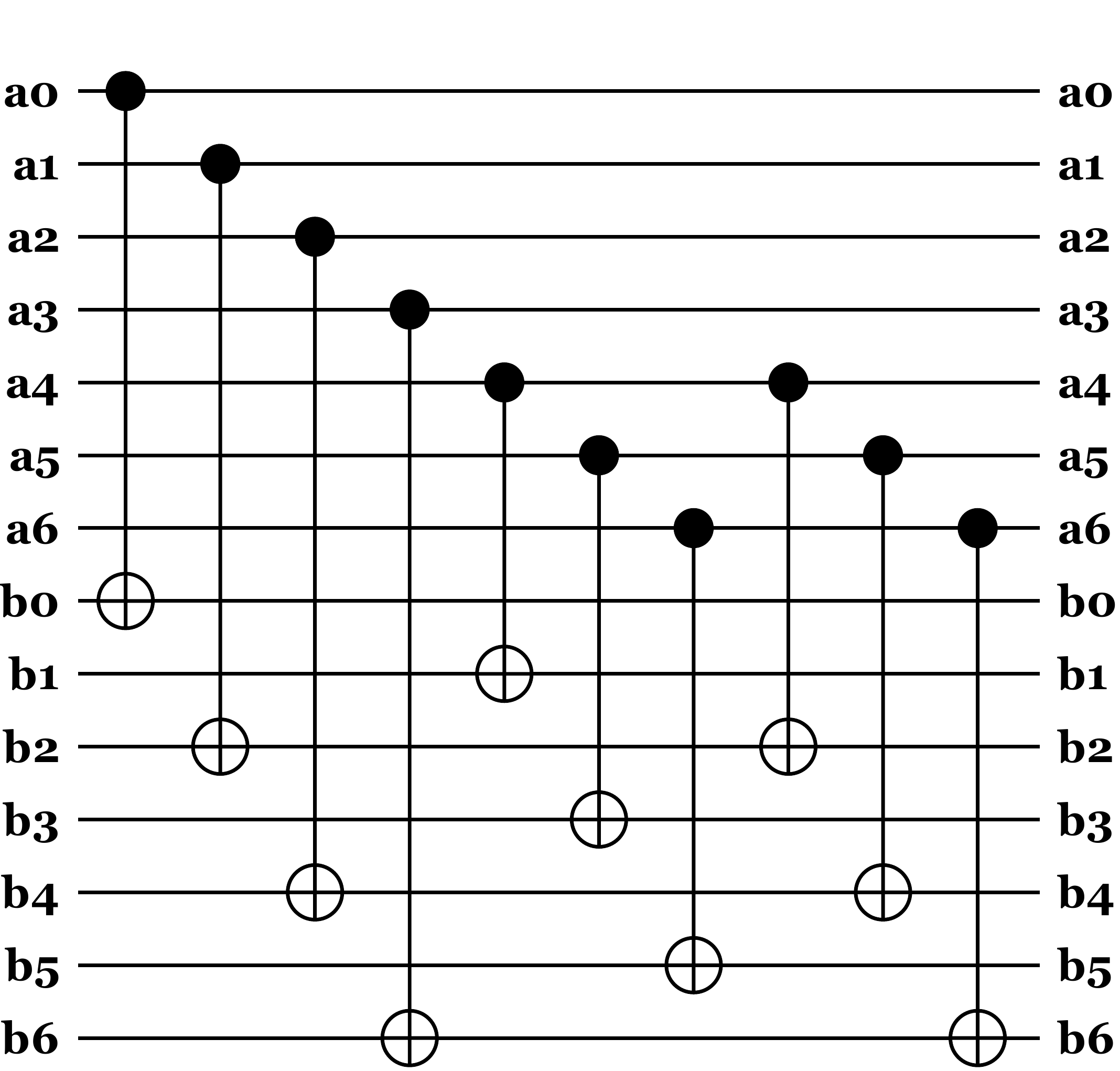}
\caption{Ancillae-free squaring of $a \in \mathbb{F}_{2}[x]/(1+x+x^7)$ in depth $2$.\label{fig:squaring}}
\end{figure}

This matrix has $10$ non-zero entries, and a maximal row or column weight of $2$. We obtain the depth $2$ circuit shown in Figure~\ref{fig:squaring} which corresponds to the following bipartite graph with chromatic index $2$:
\begin{center}
\begin{tikzpicture}[every node/.style={circle,draw}]
   \node at (0,0) (a6) {$a_6$};
   \node at (3,0) (b6) {$b_6$};
   \node at (0,1) (a5) {$a_5$};
   \node at (3,1) (b5) {$b_5$};
   \node at (0,2) (a4) {$a_4$};
   \node at (3,2) (b4) {$b_4$};
   \node at (0,3) (a3) {$a_3$};
   \node at (3,3) (b3) {$b_3$};
   \node at (0,4) (a2) {$a_2$};
   \node at (3,4) (b2) {$b_2$};
   \node at (0,5) (a1) {$a_1$};
   \node at (3,5) (b1) {$b_1$};
   \node at (0,6) (a0) {$a_0$};
   \node at (3,6) (b0) {$b_0$};

  \path (a0) edge (b0);
  \path (a4) edge (b1);
  \path (a1) edge (b2);
  \path[dashed] (a4) edge (b2);
  \path (a5) edge (b3);
  \path[dashed] (a5) edge (b4);
  \path (a2) edge (b4);
  \path (a6) edge (b5);
  \path[dashed] (a6) edge (b6);
  \path (a3) edge (b6);
\end{tikzpicture}
\end{center}

\end{example}
\begin{example}[ECDSA: squaring]
    The Digital Signature Standard \cite{FIPS1864} specifies five different fields for use in connection with binary elliptic curves along with a polynomial-basis representation for each of these fields. We used our software to find the depth and number of CNOT gates needed for an ancillae-free squaring operation with each of these representations. The corresponding values are listed in Table~\ref{tab:ecdsasquare}.
\begin{table}[htb]
\centering
\begin{tabular}{c c c }
\hline 
   \ irreducible polynomial & \ depth  & \ CNOT gates\\ 
\hline
\hline
$1+x^3+x^6+x^7+x^{163}$ & 8 & 415  \\
\hline
 $1+x^{74}+x^{233}$ & 3 & 386     \\ 
\hline
  $1+x^5+x^7+x^{12}+x^{283}$&  7 &  722 \\ 
  \hline
  $1+x^{87}+x^{409}$&3 &656 \\
  \hline
  $1+x^2+x^5+x^{10}+x^{571}$ &7 &1438\\
\hline
\hline
\end{tabular}
\caption{Resource count of an ancillae-free squaring operation for binary fields in \cite{FIPS1864}.}\label{tab:ecdsasquare}
\end{table}
\end{example}
The last two examples suggest that trinomials are an attractive choice for deriving compact ancillae-free squaring circuits, and this is indeed the case. The same holds true for computing the unique square root of an element in ${\mathbb F}_{2^n}$; the latter will be helpful for us, as the circuit used to establish Theorem~\ref{theo:middle} involves squarings as well as a square root computation for ``uncomputing''. To quantify the benefit of a ``trinomial basis representation'', first we can exploit that the irreducibility of $1+x^m+x^n\in{\mathbb F}_2[x]$ (with $m<n$) implies the irreducibility of $1+x^{n-m}+x^n\in{\mathbb F}_2[x]$ \cite[Fact~4.75]{HAC01}. So we may choose the middle-term to be of degree $\le\lfloor n/2\rfloor$. From the explicit formulae for a classical implementation by Rodr{\'\i}guez-Henr{\'\i}quez et al. \cite{RML08} we obtain the following.

\begin{proposition}\label{theo:middle}
 Let ${\mathbb F}_{2^n}={\mathbb F}_2[x]/(1+x^m+x^n)$ with $m\le\lfloor n/2\rfloor$. Then the map $\ket{a}\ket{c}\longmapsto\ket{a}\ket{c+a^2}$ (with variable input $c\in{\mathbb F}_{2^n}$) can be implemented with an ancillae-free quantum circuit of depth $\le m+1$ using no more than $3n$ CNOT gates.

Moreover, the map  $\ket{a}\ket{c}\longmapsto\ket{a}\ket{c+\sqrt{a}}$ can be implemented with an ancillae-free quantum circuit using no more than $5n$ CNOT gates.
\end{proposition}

\begin{proof}
 Let $A:=a_0+a_1x+\dots+a_{n-1}x^{n-1}$ be a representative of an ${\mathbb F}_{2^n}$-element $a$. In \cite{RML08} explicit expressions for computing the representations of $a^2$ and $\sqrt{a}$ from $a_0,\dots,a_{n-1}$ are given. Each coefficient of $a^2$ can be obtained as a sum of at most three $a_i$s. Similarly, each coefficient of $\sqrt{a}$ can be obtained as a sum of no more than five $a_i$s.

To justify the depth bound $m$ for a squaring operation, let $B:= a_0+a_1x^2+a_2x^4+ \dots + a_{n-1}x^{2n-2}$. Then $B$ is a representative of $a^2$, and the degree of $B$ is $\le 2n-2$. To find the coefficients of $a^2$, we have to find $B\bmod {x^n+x^m+1}$, i.\,e., a representative of degree less than $n$. With $\eta:=n+(n\bmod 2)$ being the smallest even number greater or equal to $n$, we can write
$$B =\underbrace{a_0+a_1x^2+\dots+a_{(\eta/2)-1} x^{\eta-2}}_{=:B_0} + \underbrace{a_{\eta/2} x^{\eta}+\dots+a_{n-1}x^{2n-2}}_{=:B_1}.$$
No reduction is needed for $B_0$, and we have
\begin{equation*}
\begin{split}
B_1& = x^n\cdot\left(a_{\eta/2} x^{\eta-n}+\dots+a_{n-1}x^{n-2}\right)\\
& =(1+x^m)\cdot\left(a_{\eta/2}x^{\eta-n}+\dots+a_{n-1}x^{n-2}\right)\\
& = \underbrace{a_{\eta/2}x^{\eta-n}+\dots+a_{n-1}x^{n-2}}_{=:B_{10}} + \underbrace{a_{\eta/2}x^{\eta+m-n}+\dots+a_{n-1}x^{m+n-2}}_{=:B_{11}}
\end{split}\quad.
\end{equation*}
No reduction is needed for $B_{10}$, and we can compute $B_0 + B_{10}$ in depth $2-(n \bmod 2)$.
We can reduce $B_{11}$, a polynomial of degree $\le m+n-2$, in the same way as we just did with $B_1$, and after at most $m-1$ reduction steps we obtain a representative of degree less than $n$. This increases the circuit depth at most by $m-1$, resulting in a total depth of at most $(m-1)+2-(n\bmod 2)\le m+1$.
\end{proof}

\section{Adding a fixed point with reduced $T$-gate complexity}\label{sec:ellipticcurve}
All of the binary elliptic curves proposed in the Digital Signature Standard \cite{FIPS1864} fall in the class of so-called ordinary binary elliptic curves. In general, such curves can be expressed by means of a short Weierstra{\ss} equation
\begin{equation}y^2 +xy = x^3 + a_2x^2 + a_6\label{equ:weierstrass}
\end{equation}
where $a_2, a_6 \in \mathbb{F}_{2^n}$ with $a_6 \neq 0$. We write $$E_{a_2,a_6}(\mathbb{F}_{2^n}) := \{ (x,y)\in \mathbb{F}_{2^n}: y^2 +xy= x^3 + a_2x^2+a_6\} \cup \{ \mathcal O\}$$
for the set of (${\mathbb F}_{2^n}$-rational) points on such a curve. The (projective) point $\mathcal O$ is often referred to as \emph{point at infinity} and serves as neutral element in the group $E_{a_2,a_6}(\mathbb{F}_{2^n})$. With this affine representation of an ordinary binary elliptic curve, the group law is summarized in the following Algorithm~\ref{algo:solinas}, taken from \cite{Sol98}. At this $P_1=(x_1, y_1)\in E_{a_2,a_6}(\mathbb{F}_{2^n})$ and $P_2=(x_2, y_2)\in E_{a_2,a_6}(\mathbb{F}_{2^n})$.\footnote{As $(0,0)\not\in E_{a_2,a_6}(\mathbb{F}_{2^n})$, the neutral element $\mathcal O$ can be represented as $(0,0)$.}

\begin{algorithm}[h]
  \KwData{Points $P_1=(x_1, y_1)$ and $P_2=(x_2,y_2)$ on $E_{a_2,a_6}(\mathbb{F}_{2^n})$.}
  \KwResult{Point $P_3=(x_3, y_3)$ with $P_3=P_1+P_2$.}
{\If{$P_1=\mathcal O$}{return $P_2$}
\If{$P_2=\mathcal O$}{return $P_1$}
\eIf{$x_1=x_2$}{\eIf(\tcc*[f]{$P_1=-P_2$}){$y_1+y_2=x_2$}{\Return{$\mathcal O$}}
(\tcc*[f]{$P_1=P_2$}){$m = x_2+y_2/x_2$\;
 $x_3 = m^2 + m + a_2$\;
 $y_3=x_{2}^{2} + (m+1)\cdot x_3$}}(\tcc*[f]{$P_1\ne\pm P_2$})
 {$m=(y_1+y_2)/(x_1+x_2)$\;
  $x_3=m^2+m+x_1+x_2+a_2$\;
 $y_3=(x_2+x_3)\cdot m+x_3+y_2$}
\Return {$(x_3,y_3)$}}
\caption{Adding two points on an ordinary binary elliptic curve using affine coordinates.}\label{algo:solinas}
\end{algorithm}

Kaye and Zalka \cite{KaZa04} argue that to implement Shor's algorithm it it sufficient to provide a quantum circuit that implements the ``generic branch'' $P_1\ne\pm P_2$ of Algorithm~\ref{algo:solinas} for a fixed point $P_2$, and we restrict to this situation. To avoid the (costly) inversion operation, one usually implements this point addition in a projective representation. The standard projective representation $(X,Y,Z)\in{\mathbb F}_{2^n}\setminus\{(0,0,0)\}$ of an affine point $(x,y)$ satisfies $x=X/Z$ and $y=Y/Z$. Here we follow a different convention, introduced by L\'opez and Dahab \cite{LoDa98}, that has also been used for the addition circuit in \cite{ARS12b}: the affine point $(x,y)$ is represented projectively by $(X,Y,Z)$ with $x=X/Z$ and $y= Y/Z^2$. Accordingly, the curve given by Equation~\eqref{equ:weierstrass} would be expressed as
\begin{equation}
Y^2+XYZ= X^3Z+a_2X^2Z^2+a_6Z^4,\label{equ:projective}
\end{equation}
the identity element $\mathcal O$ being represented by $(X,0,0) \in \mathbb{F}_{2^n}^3\setminus\{(0,0,0)\}$. Based on an addition formula by Higuchi and Takagi \cite{HiTa00} for this type of projective representation, in \cite{ARS12b} the following result is given, where
\begin{itemize}
   \item $G_M(n)$ and $D_M(n)$ denote the number of gates and depth needed to implement an ${\mathbb F}_{2^n}$-multiplier, respectively;
 \item $G_M^T(n)$ and $D_M^T(n)$ denote the number of $T$-gates and $T$-depth needed to implement an ${\mathbb F}_{2^n}$-multiplier, respectively.
\end{itemize}

\begin{proposition}[{\cite[Proposition~3.2]{ARS12b}}]\label{prop:oldrecord}
   Let $P_2$ be a fixed point on $E_{a_2,a_6}(\mathbb{F}_{2^n})$. With the above-mentioned variant of projective coordinates, the addition $\ket{X_1}\ket{Y_1}\ket{Z_1}\ket{0}\ket{0}\ket{0}\longmapsto\ket{X_1}\ket{Y_1}\ket{Z_1}\ket{X_3}\ket{Y_3}\ket{Z_3}$ can be carried out with a quantum circuit $\mathcal C$ satisfying all of the following:
\begin{itemize}
   \item The total number of $T$-gates in $\mathcal C$ is $13\cdot G_M^T(n)$.
   \item The total number of gates in $\mathcal C$ is at most $13\cdot G_M(n)$ plus $12n^2+\bigO(n)$ (the latter being CNOT gates).
   \item The $T$-depth of $\mathcal C$ is $4\cdot D_M^T(n)$.
   \item The overall depth of $\mathcal C$ is $4\cdot D_M(n)$ plus $4n+\bigO(1)$ (the latter being CNOT gates).
\end{itemize}
This includes the cost of cleaning up ancillae. If $(X_1,Y_1,Z_1)$ is not the identity or equal to $\pm P_2$, then $(X_3, Y_3, Z_3)$ is a representation of the sum of $(X_1, Y_1, Z_1)$ and the fixed point $P_2$ in the above-mentioned variant of projective coordinates.
\end{proposition}
To the best of our knowledge, in terms of $T$-gate complexity this is currently the most efficient quantum circuit that has been published for the ``generic addition'' of a fixed point on an ordinary binary elliptic curve.

\subsection{An addition circuit based on a formula by Al-Daoud et al.}\label{addition}
Invoking L\'opez-Dahab coordinates as described above, in \cite{AMRK02} Al-Daoud et al. present a point addition formula which seems well suited for a quantum circuit that aims at adding a fixed point.
Besides requiring only four general multiplications, in two cases a constant multiplication and a squaring operation can naturally be combined into a single matrix-vector multiplication. More specifically, let $P_2=(x_2, y_2,1)$ be a fixed point on the curve given by Equation~\eqref{equ:projective}, and let $P_1=(X_1,Y_1,Z_1)$ be an arbitrary point on this curve (which will be given as input to our quantum circuit). We assume that ${\mathcal O}\ne P_1\ne\pm P_2$. Then a representation $(X_3, Y_3, Z_3)$ of the sum $P_1+P_2$ can be computed as follows.
$$\begin{array}{ccl ccl ccl}
   A&=& Y_1 + y_2Z_1^2, &B &=& X_1 + x_2Z_1,& C&=&B\cdot Z_1,\\
   Z_3&=&C^2, &D&=&x_2Z_3,\\
   X_3&=&A^2 + C\cdot(A+B^2+a_2C),\\Y_3&=&(D+X_3)\cdot(A\cdot C+Z_3)+(y_2+x_2)Z_3^2
\end{array}$$
The above formulation is taken from the \emph{explicit-formulas database} \cite[\emph{madd-2005-dl}]{EFDB12} (see also \cite[Chapter~13.3.1.d]{CoFr06}). To characterize the complexity of our addition circuit, it is appropriate to distinguish between the resources for general multiplication, squaring, and other matrix-vector multiplications. As manifested in Proposition~\ref{theo:middle} and Table~\ref{tab:ecdsasquare}, for certain field representations the resource count of a squaring operation is remarkably modest, even for cryptographically significant field sizes. So in the sequel we write $G_S(n)\le n^2-n+1$ for the number of (CNOT) gates needed to implement a squaring operation with the underlying representation of ${\mathbb F}_{2^n}$ and analogously $D_S(n)\le n$ for the depth of such a circuit. The number of qubits needed in our construction will depend on the details of the underlying ${\mathbb F}_{2^n}$-multiplier. So to quantify the number of qubits, we assume that the multiplication of any $a,b\in{\mathbb F}_{2^n}$---i.\,e., the function $\ket{a}\ket{b}\ket{c}\mapsto\ket{a}\ket{b}\ket{c+a\cdot b}$ with $c\in{\mathbb F}_{2^n}$ arbitrary---is realized with $$\underbrace{n+n}_{\text{input}}+\underbrace{n}_{\text{output}}+\underbrace{A_M(n)}_{\text{ancillae}}$$ qubits. With this notation we obtain the following.

\begin{figure}[tbh]
\center
\includegraphics[width=0.95\textwidth]{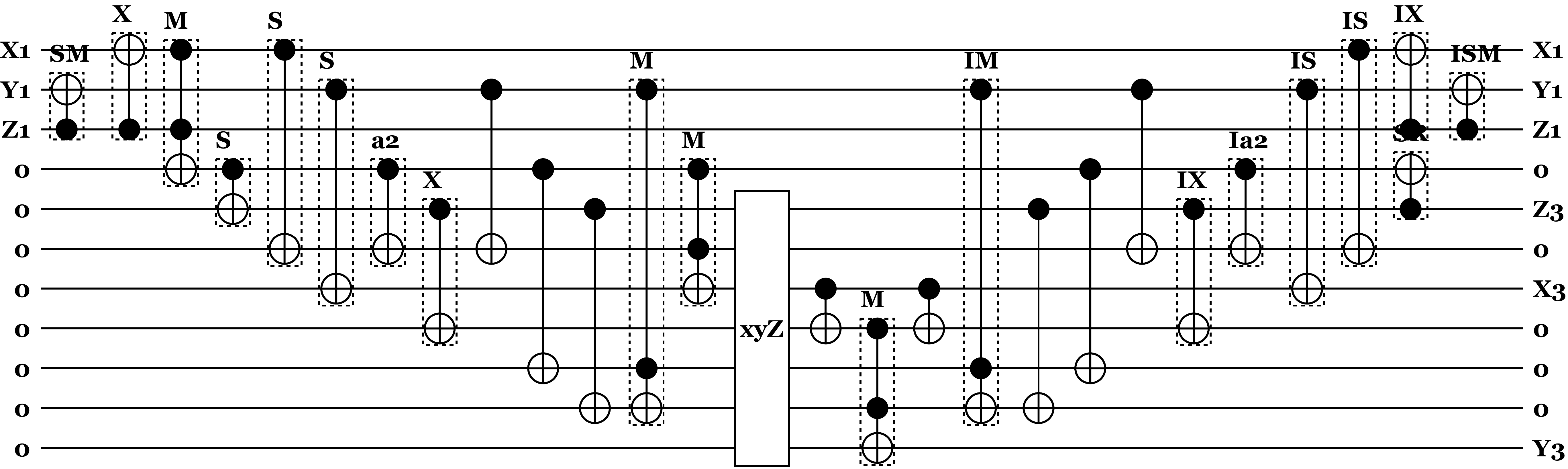}
\caption{A complete circuit adding a point in L{\'o}pez-Dahab coordinates on $E_{1,1}(\mathbb{F}_{2})$ with the fixed affine point $(1,1)$. Parsing the circuit from left to right, the initial gates labeled {\bf SM}, ${\bf X}$, and {\bf M} correspond to the operations in Steps~1--3 of the proof of Theorem~\ref{theo:main}. The subsequent three gates (labeled {\bf S}) implement the parallel squarings in Step~4. This is followed by the parallel scalar multiplications with $a_2$ ({\bf a2}) and $x_2$ ({\bf X}) from Step~5 and three CNOT gates operating on disjoint wires to implement Step~6. The two multipliers in Step~7 are realized by two Toffoli gates (marked {\bf M}), and for the sake of completeness we include a box labelled {\bf xyZ} which is actually the identity as for our specific example the value of $x_2+y_2$ in Step~7 is $0$. Step~8 corresponds to a single CNOT gate, and the subsequent Toffoli gate labelled {\bf M} implements Step~9. Starting the clean-up part of the circuit, the CNOT from Step~10 is used, followed by the reversal of a multiplier in Step~11 ({\bf IM}). Step~12 results in three CNOT gates. This is followed by the reversal of the scalar multiplications in Step~13 ({\bf IX} and {\bf Ia2}). Reversing of the squaring operations from Step~14 is implemented by the two gates marked {\bf IS}. They can be executed in parallel with the gate marked {\bf SR}, realizing the square root computation in Step~14. Eventually, Step~15 corresponds to the gate labelled {\bf IX}, and Step~16 is realized by a singe CNOT gate ({\bf ISM}).}\label{fig:circuit}
\end{figure}

\begin{theorem}\label{theo:main}
Let $P_2$ be a fixed point on the curve $E_{a_2,a_6}(\mathbb{F}_{2^n})$. Using L{\'o}pez-Dahab coordinates, the addition $\ket{X_1}\ket{Y_1}\ket{Z_1}\ket{0}\ket{0}\ket{0}\longmapsto\ket{X_1}\ket{Y_1}\ket{Z_1}\ket{X_3}\ket{Y_3}\ket{Z_3}$ of this point can be carried out with a quantum circuit $\mathcal C$ satisfying all of the following:
\begin{itemize}
   \item The total number of $T$-gates in $\mathcal C$ is $5G_M^T(n)$.
   \item The total number of gates in $\mathcal C$ is at most $5G_M(n)$ plus $5G_S(n)+10n^2-2n+10$ (the latter being CNOT gates).
   \item The $T$-depth of $\mathcal C$ is $4D_M^T(n)$.
   \item The overall depth of $\mathcal C$ is $3D_M(n)+\max(D_M(n),n)$ plus $D_S(n)+7n+4$ (the latter being CNOT gates).
   \item The total number of qubits, including the $3n$ qubits for storing the input, is $11n+4A_M(n)$.
\end{itemize}
This includes the cost to clean up ancillae. If $(X_1,Y_1,Z_1)$ is not the identity or equal to $\pm P_2$, then $(X_3, Y_3, Z_3)$ is a representation in L{\'o}pez-Dahab coordinates of the sum of $(X_1, Y_1, Z_1)$ with~$P_2$.
\end{theorem}
\begin{proof}
To find $X_3$, $Y_3$ and $Z_3$ we proceed as follows.

\begin{enumerate}\item Using $\le (n^2-n+1)$-CNOT gates and depth $\le n$, we can compute $A=Y_1+y_2 \cdot Z_{1}^2$, i.\,e., the wires originally storing $Y_1$ now store $A$.
\item Similarly, we now compute $B=X_1+x_2\cdot Z_1$ using $\le n^2-n+1$ CNOT gates in depth $\le n$, storing $B$ in the wires originally holding $X_1$.
\item \label{step:multip}Multiplying $B$ and $Z_1$ we store $C=B\cdot Z_1$ into a new set of $n$ wires (initialized to $\ket{0}$). This increases the depth by $D_M(n)$ and uses $G_M(n)$ gates. Similarly the $T$-depth and number of $T$-gates are increased by $D_M^T(n)$ and $G_M^T(n)$, respectively.
\item Square $C$, $A$ and $B$ to obtain $Z_3$, $A^2$ and $B^2$ in parallel in depth $D_S(n)$ by using $3G_S(n)$-CNOT gates. To store the results we add $3n$ additional ($\ket{0}$-initialized) qubits. The wires holding $A^2$ will be used to store $X_3$.
\item \label{step:twicelinear}Now compute $a_2 \cdot C$ in depth $\le n$ using $\le n^2-n+1$-CNOT gates. The result of this operation is added directly to $B^2$. At the same time, compute $D=x_2\cdot Z_3$  using another $\le n^2-n+1$ CNOT gates. The latter result is stored in $n$ new ($\ket{0}$-initialized) qubits.
\item \label{step:copies}Add $A$ to the qubits holding $B^2+a_2 \cdot C$. For this, $n$ CNOT gates and depth $1$ suffice. Simultaneously we can apply $n+n$ more CNOT gates to create a copy $C'$ of $C$ and a copy $Z_3'$ of $Z_3$ in a set of $n+n$ new ($\ket{0}$-initialized) qubits. Having $C'$ available enables us to perform the next two multiplications in parallel.
\item With $\le 2G_M(n)+n^2-n+1$ gates and increasing the depth by $\le \max(D_M(n),n)$, we can now compute in parallel      \begin{itemize}
         \item $C\cdot (A+B^2+a_2\cdot C)$ and add the result onto the wires for the value $X_3$;
         \item $A \cdot C'$ and add the result onto the wires holding $Z_3'$;
         \item $(x_2+y_2)\cdot Z_3^2$ and store the result in the ($\ket{0}$-initialized) wires for the value $Y_3$.
\end{itemize}
This step increases the $T$-depth by $D_M^T(n)$ and the number of $T$-gates by $2G_M^T(n)$.
\item With $n$ CNOT gates in depth~1 we can add $X_3$ to the wires storing $D$.
\item Find $(D+X_3) \cdot (A \cdot C'+ Z_3')$ using $G_M(n)$ gates which increases the depth by $D_M(n)$ units. The result is added to the wires which are to hold $Y_3$. This step increases the number of $T$-gates by $G_M(n)$ and the $T$-depth by $D_M^T$.

At this point we have computed all of $X_3$, $Y_3$ and $Z_3$, and we are left with cleaning up ancillae and restoring the input values.

\item Add $X_3$ to the wires holding $D+X_3$ for which we need $n$ CNOT gates and which increases the circuit depth by $1$.
\item Reversing the multiplication $A\cdot C'$ takes depth $D_M(n)$ and requires $G_M(n)$ gates. This also increases the $T$-depth by $D_M^T(n)$ and the number of $T$-gates accordingly by $G^T_M(n)$.
\item To reverse the CNOT operations from Step~\ref{step:copies} we execute them again in depth $1$, using $3n$ CNOT gates.
\item Next, the two linear operations from Step~\ref{step:twicelinear} can be run backwards simultaneously, increasing the depth by $\le n$ and adding $\le 2\cdot(n^2-n+1)$ CNOT gates to the total gate count.
\item The squarings of $A$ and $B$ can be run backwards simultaneously using $2G_S(n)$ gates. Simultaneously we can apply a square root computation to $Z_3=C^2$ to cancel the output $C$ of the multiplier in Step~\ref{step:multip}. The square root computation can be done with $\le n^2-n+1$ CNOT gates, and with $D_S(n)\le n$ we see that the overall depth of this step is $\le n$.
\item Reversing the computation of $B$ takes $\le n^2-n+1$ CNOT gates and can be completed in depth $\le n$.
\item Finally, reversing the computation of $A$ increases the gate count by $\le n^2-n+1$ CNOT gates and the depth by $n$.
 \end{enumerate}
Table~\ref{tab:resources} summarizes the resource count for each of these steps. In the column for the number of qubits we count all qubits that are used on top of the $3n$~qubits necessary to represent the input $(X_1,Y_1,Z_1)$; this includes the $3n$ bits needed to store the result $(X_3, Y_3, Z_3)$. Exploiting that the multipliers are the only parts of the circuit involving $T$-gates, from this table we immediately obtain the bounds claimed.
\begin{table}[htb]
\begin{center}
\begin{tabular}{|r|c|c|c|}
\hline
step& no. gates& depth & no. qubits\\
\hline\hline
1&$n^2-n+1$&$n$&$0$\\
\hline
2&$n^2-n+1$&$n$&$0$\\
\hline
3&$G_M(n)$&$D_M(n)$&$n+A_M(n)$\\
\hline
4&$3G_S(n)$&$D_S(n)$&$3n$\\
\hline
5&$2\cdot(n^2-n+1)$&$n$&$n$\\
\hline
6&$3n$&$1$&$2n$\\
\hline
7&$2G_M(n) + n^2-n+1$&$\max(D_M(n),n)$&$n+2A_M(n)$\\
\hline
8&$n$&$1$&$0$\\
\hline
9&$G_M(n)$&$D_M(n)$&$A_M(n)$\\
\hline\hline
10&$n$&$1$&$0$\\
\hline
11&$G_M(n)$&$D_M(n)$&$0$\\
\hline
12&$3n$&$1$&$0$\\
\hline
13&$2\cdot(n^2-n+1)$&$n$&$0$\\
\hline
14&$2\cdot G_S(n)+n^2-n+1$&$n$&$0$\\
\hline
15&$n^2-n+1$&$n$&$0$\\
\hline
16&$n^2-n+1$&$n$&$0$\\
\hline
\end{tabular}
\end{center}
\caption{Resource bounds for each step of the circuit in the proof of Theorem~\ref{theo:main}.}\label{tab:resources}
\end{table}

\end{proof}

\begin{remark} Proposition~\ref{prop:oldrecord} does not give an explicit count for the number of qubits, but the proof of {\cite[Proposition~3.2]{ARS12b}} emphasizes parallelization. Step~1--3 of the latter already add $6n$ new wires to the $3n$~qubits for the input. Step~4 then executes $4$ field multiplications in parallel (invoking $4A_M(n)$ ancillae), and Step~6 runs three more multipliers in parallel, storing the result in $3n$ new wires, so it is fair to conclude that the total number of qubits is larger than the bound $11n+4A_M(n)$ established in Theorem~\ref{theo:main}.

Also, it is worth noting that the resource bounds in Theorem~\ref{theo:main} are indeed worst-case bounds. In cryptographic applications it is common to choose $a_2\in\{0,1\}$, thereby eliminating the need to implement the ($T$-gate free) computation of $a_2\cdot C$.
\end{remark}
%

The proof of Theorem~\ref{theo:main} is constructive, and we implemented a software tool which for a given irreducible polynomial $p\in{\mathbb F}_2[x]$, a curve point $Q\in E_{a_2,a_6}({\mathbb F}_2[x]/(p))$ and $a_2\in{\mathbb F}_2[x]/(p)$ generates the corresponding  quantum circuit for the ``generic addition'' of $Q$. As programming language we chose {\tt Python}, and the resulting quantum circuits are stored in a text file using the {\tt .qc} format. This format supports the grouping of gates into subcircuits, allowing a user to hide the details of, e.\,g., a field multiplier, when viewing the circuit in {\tt QCViewer}. Figure~\ref{fig:circuit} gives an example of a complete addition circuit using the curve $E_{1,1}(\mathbb{F}_{2})$ and $P=(1,1)$ as fixed point to be added. The ${\mathbb F}_2$-multiplier is realized as a Toffoli gate, requiring $A_M(1)=0$ ancillae. As detailed by Amy et al. in \cite{AMMR13}, a Toffoli gate can be decomposed into a circuit involving a total of $G_M(1)=15$~gates, $G_M^T(1)=7$ of which are $T$-gates and the remaining ones being CNOT and Hadamard gates. This can be done with a $T$-depth of $D_M^T(1)=4$ and an overall depth of $D_M(1)=8$.

By means of our software, we also experimented with larger curves. For larger curves, the detailed $T$-gate complexity of the circuit depends very much on the complexity of the underlying ${\mathbb F}_{2^n}$-multiplier. For our experiments we built on an existing {\tt Python} code by Brittanney Amento to produce a {\tt .qc} description of an ${\mathbb F}_{2^n}$-multiplier. Our software treats the multiplier basically as a black box, however. So if improved quantum circuits for ${\mathbb F}_{2^n}$-multiplication become available, integrating them with the existing code should not be a problem.

As final example, we take a look the square root computation (see Step~14 in the proof of Theorem~\ref{theo:main}) for binary fields in the Digital Signature Standard: 
\begin{example}[ECDSA: square root computation] Table~\ref{tab:ecdsasqrt} lists depth and gate counts for the ancillae-free square root computation for the binary fields in \cite{FIPS1864}. As can be seen, for the case of a ``trinomial basis'' this operation can be implemented quite efficiently.
\begin{table}[htb]
\centering
\begin{tabular}{c c c }
\hline 
   \ irreducible polynomial & \ depth  & \ CNOT gates\\ 
\hline
\hline
$1+x^3+x^6+x^7+x^{163}$ & 104 & 7399  \\
\hline
 $1+x^{74}+x^{233}$ & 6 & 591     \\ 
\hline
  $1+x^5+x^7+x^{12}+x^{283}$&  94 &  11657 \\ 
  \hline
  $1+x^{87}+x^{409}$&2 & 613 \\
  \hline
  $1+x^2+x^5+x^{10}+x^{571}$ &273 &76172\\
\hline
\hline
\end{tabular}
\caption{Resource count of an ancillae-free square root computation for binary fields in \cite{FIPS1864}.}\label{tab:ecdsasqrt}
\end{table}
\end{example}

\section{Conclusion}
The presented quantum circuit for point addition reduces an important cost parameter over the best previous solution---the number of $T$-gates can be reduced by more than 60\% without increasing $T$-depth. At the same time, the number of qubits can be reduced. The overall depth increases linearly, but in view of the savings achieved the depth increase looks acceptable. Aiming at the implementation of elliptic curve arithmetic for cryptanalytic applications, the ability to synthesize (optimized) point addition circuits automatically seems very helpful. We also hope that the concrete complexity bounds provided along with the capability to derive actual circuits in an established format simplifies quantitative comparisons and stimulates follow-up research on more efficient implementations.

\paragraph{Acknowledgments.} The authors thank Stephen Locke for helpful discussions on graph coloring and Brittanney Amento for kindly allowing us to use her {\tt Python} code to generate quantum circuits for ${\mathbb F}_{2^n}$-multiplication. RS is supported by NATO's Public Diplomacy Division in the framework of ``Science for Peace'', Project MD.SFPP 984520.


\begin{thebibliography}{10}

\bibitem{AMRK02}
Essame Al-Daoud, Ramlan Mahmod, Mohammad Rushdan, and Adem Kilicman.
\newblock {A New Addition Formula for Elliptic Curves over $\mathrm{GF}(2^n)$}.
\newblock {\em IEEE Transactions on Computers}, 51(8):972--975, August 2002.

\bibitem{ARS12b}
Brittanney Amento, Martin R{\"o}tteler, and Rainer Steinwandt.
\newblock {Efficient quantum circuits for binary elliptic curve arithmetic:
  reducing $T$-gate complexity}.
\newblock {\em Quantum Information \& Computation}, 13:631--644, July 2013.

\bibitem{ARS12}
Brittanney Amento, Martin R{\"o}tteler, and Rainer Steinwandt.
\newblock {Quantum binary field inversion: improved circuit depth via choice of
  basis representation}.
\newblock {\em Quantum Information \& Computation}, 13:116--134, January 2013.

\bibitem{AMMR13}
Matthew Amy, Dmitri Maslov, Michele Mosca, and Martin Roetteler.
\newblock {A Meet-in-the-Middle Algorithm for Fast Synthesis of Depth-Optimal
  Quantum Circuits}.
\newblock {\em IEEE Transactions on Computer-Aided Design of Integrated
  Circuits and Systems}, 32(6):818--830, June 2013.
\newblock For a preprint version see \cite{AMMR13b}.

\bibitem{AMMR13b}
Matthew Amy, Dmitri Maslov, Michele Mosca, and Martin Roetteler.
\newblock {A meet-in-the-middle algorithm for fast synthesis of depth-optimal
  quantum circuits}.
\newblock arXiv:quant-ph/1206.0758v3, January 2013.
\newblock Available at \url{http://arxiv.org/abs/1206.0758v3}.

\bibitem{BBF03}
St{\'e}phane Beauregard, Gilles Brassard, and Jos{\'e}~M. Fernandez.
\newblock {Quantum Arithmetic on Galois Fields}.
\newblock arXiv:quant-ph/0301163v1, January 2003.
\newblock Available at \url{http://arxiv.org/abs/quant-ph/0301163v1}.

\bibitem{EFDB12}
Daniel~J. Bernstein and Tanja Lange.
\newblock Explicit-formulas database.
\newblock \url{http://www.hyperelliptic.org/EFD/index.html}.

\bibitem{BLF08}
Daniel~J. Bernstein, Tanja Lange, and Reza~Rezaeian Farashahi.
\newblock {Binary Edwards Curves}.
\newblock In Elisabeth Oswald and Pankaj Rohatgi, editors, {\em Cryptographic
  Hardware and Embedded Systems -- CHES 2008}, volume 5154 of {\em Lecture
  Notes in Computer Science}, pages 244--265. International Association for
  Cryptologic Research, Springer, 2008.

\bibitem{CoFr06}
Henri Cohen and Gerhard Frey, editors.
\newblock {\em {Handbook of Elliptic and Hyperelliptic Curve Cryptography}}.
\newblock Discrete mathematics and its applications. Chapman \& Hall/CRC, 2006.

\bibitem{COS01}
Richard Cole, Kirstin Ost, and Stefan Schirra.
\newblock {Edge-coloring bipartite multigraphs in $O(E\log D)$ time}.
\newblock {\em Combinatorica}, 21(1):5--12, 2001.

\bibitem{qcreference}
Institute for Quantum~Computing.
\newblock {QCViewer}.
\newblock \url{http://qcirc.iqc.uwaterloo.ca/index.php?n=Projects.QCViewer},
  2013.

\bibitem{Pythonreference}
Python~Software Foundation.
\newblock {Python Programming Language -- Official Website}.
\newblock \url{http://www.python.org}, 2013.

\bibitem{HiTa00}
Akira Higuchi and Naofumi Takagi.
\newblock {A fast addition algorithm for elliptic curve arithmetic using
  projective coordinates}.
\newblock {\em Information Processing Letters}, 76:101--103, 2000.

\bibitem{KaZa04}
Phillip Kaye and Christof Zalka.
\newblock {Optimized quantum implementation of elliptic curve arithmetic over
  binary fields}.
\newblock arXiv:quant-ph/0407095v1, July 2004.
\newblock Available at \url{http://arxiv.org/abs/quant-ph/0407095v1}.

\bibitem{LoDa98}
Julio L{\'o}pez and Ricardo Dahab.
\newblock {Improved Algorithms for Elliptic Curve Arithmetic in $GF(2^n)$}.
\newblock In Stafford Tavares and Henk Meijer, editors, {\em Selected Areas in
  Cryptography -- SAC'98}, volume 1556 of {\em Lecture Notes in Computer
  Science}, pages 201--212. Springer, 1999.

\bibitem{MMCP09b}
Dmitri Maslov, Jimson Mathew, Donny Cheung, and Dhiraj~K. Pradhan.
\newblock {An $O(m^2)$-depth quantum algorithm for the elliptic curve discrete
  logarithm problem over GF$(2^m)$}.
\newblock {\em Quantum Information \& Computation}, 9(7):610--621, 2009.
\newblock For a preprint version see \cite{MMCP09}.

\bibitem{MMCP09}
Dmitri Maslov, Jimson Mathew, Donny Cheung, and Dhiraj~K. Pradhan.
\newblock {On the Design and Optimization of a Quantum Polynomial-Time Attack
  on Elliptic Curve Cryptography}.
\newblock arXiv:0710.1093v2, February 2009.
\newblock Available at \url{http://arxiv.org/abs/0710.1093v2}.

\bibitem{HAC01}
Alfred~J. Menezes, Paul~C. van Oorschot, and Scott~A. Vanstone.
\newblock {\em {Handbook of Applied Cryptography}}.
\newblock CRC Press, August 2001.
\newblock Sample chapters available at \url{http://cacr.uwaterloo.ca/hac/}.

\bibitem{FIPS1864}
National Institute of Standards and Technology, Gaithersburg, MD 20899-8900.
\newblock {\em FIPS PUB 186-4. Federal Information Processing Standard
  Publication. Digital Signature Standard (DSS)}, July 2013.
\newblock Available at
  \url{http://nvlpubs.nist.gov/nistpubs/FIPS/NIST.FIPS.186-4.pdf}.

\bibitem{Poi06}
Alain Pointdexter.
\newblock {edge-coloring of a bipartite graph (Python recipe)}.
\newblock Available at
  \url{http://code.activestate.com/recipes/498092-edge-coloring-of-a-bipartite-graph/},
  September 2013.

\bibitem{RML08}
Francisco Rodr{\'{\i}}guez-Henr{\'{\i}}quez, Guillermo Morales-Luna, and Julio
  L{\'o}pez.
\newblock {Low-Complexity Bit-Parallel Square Root Computation over $GF(2^m)$
  for all Trinomials}.
\newblock {\em IEEE Transactions on Computers}, 57(4):472--480, April 2008.
\newblock For a preprint version see \cite{RML08b}.

\bibitem{RML08b}
Francisco Rodr{\'{\i}}guez-Henr{\'{\i}}quez, Guillermo Morales-Luna, and Julio
  L{\'o}pez-Hern{\'a}ndez.
\newblock {Low Complexity Bit-Parallel Square Root Computation over $GF(2^m)$
  for all Trinomials}.
\newblock Cryptology ePrint Archive: Report 2006/133, April 2006.
\newblock Available at \url{http://eprint.iacr.org/2006/133}.

\bibitem{RoSt13}
Martin R{\"o}tteler and Rainer Steinwandt.
\newblock {A quantum circuit to find discrete logarithms on ordinary binary
  elliptic curves in depth $\mathrm{O}(\log^2 n)$}.
\newblock {\em Quantum Information \& Computation}, (accepted; to appear).

\bibitem{Sho97}
Peter~W. Shor.
\newblock {Polynomial-Time Algorithms for Prime Factorization and Discrete
  Logarithms on a Quantum Computer}.
\newblock {\em SIAM Journal on Computing}, 26(5):1484--1509, 1997.

\bibitem{Sol98}
Jerome~A. Solinas.
\newblock {An Improved Algorithm for Arithmetic on a Family of Elliptic
  Curves}.
\newblock In Burton S.~Kaliski Jr., editor, {\em Advances in Cryptology --
  CRYPTO '97}, volume 1294 of {\em Lecture Notes in Computer Science}, pages
  357--371. Springer, 1997.

\end{thebibliography}
\end{document}